\documentclass[12pt]{article}
\usepackage{newpxtext}
\usepackage{amsthm,amsmath,amssymb}
\usepackage[normalem]{ulem}
\usepackage[bookmarks=true,hypertexnames=false,pagebackref]{hyperref}
\hypersetup{colorlinks=true, citecolor=blue, linkcolor=red, urlcolor=blue}
\usepackage{thm-restate}
\usepackage{cleveref}
\usepackage{fullpage}
\usepackage{bm}
\usepackage[ruled,linesnumbered]{algorithm2e}
\usepackage{enumerate}
\usepackage{xcolor}

\newtheorem{theorem}{Theorem}[section]
\newtheorem{claim}[theorem]{Claim}

\newtheorem{lemma}[theorem]{Lemma}
\newtheorem{definition}[theorem]{Definition}
\newtheorem{conjecture}[theorem]{Conjecture}

\newtheorem{observation}[theorem]{Observation}

\def\E{\mathbb{E}}
\def\bits{\{0,1\}}
\newcommand{\F}[1]{\mathcal{F}^{#1}}
\newcommand{\Fsize}[2]{\mathcal{F}_{\text{size}}^{#1}[#2]}
\newcommand{\disprob}[2]{\mu_{\text{prob}}^{#1}[#2]}
\newcommand{\dissize}[2]{\mu_{\text{size}}^{#1}[#2]}
\newcommand{\s}{\mathfrak{s}}

\newcommand{\cF}{\mathcal{F}}
\newcommand{\rv}[1]{\bm{#1}}    
\newcommand{\rX}{\rv{X}}
\newcommand{\rY}{\rv{Y}}
\newcommand{\rx}{\rv{x}}
\newcommand{\ry}{\rv{y}}
\newcommand{\rz}{\rv{z}}
\newcommand{\rW}{\rv{W}}
\newcommand{\rj}{\rv{j}}
\newcommand{\rd}{\rv{d}}
\newcommand{\rU}{\rv{U}}
\newcommand{\re}{\rv{e}}
\newcommand{\rV}{\rv{V}}
\newcommand{\rR}{\rv{R}}
\newcommand{\rb}{\rv{b}}

\newcommand{\rpi}{\rv{\pi}}

\newcommand{\rSigma}{\rv{\Sigma}}
\def\locs{\text{locs}}
\def\next{\text{next}}

\newcommand{\val}{\text{val}}

\newcommand{\unif}{\text{Uniform}}

\newcommand{\planted}{\text{Planted}}

\newcommand{\set}{\text{Sets}}
\newcommand{\size}{\text{Sizes}}

\newcommand{\alg}{\mathcal{ALG}}
\newcommand{\one}[1]{\textbf{1}[#1]}

\title{Streaming Lower Bounds and Asymmetric Set-Disjointness}

\begin{document}
\author{
Shachar Lovett
\thanks{Research supported by NSF awards 1953928 and 2006443.}\\
Computer Science Department\\
University of California San Diego\\
\texttt{slovett@ucsd.edu}
\and
Jiapeng Zhang
\thanks{Research supported by NSF CAREER award 2141536.}
\\
Department of Computer Science\\
University of Southern California\\
\texttt{jiapengz@usc.edu}
}
\maketitle

\begin{abstract}
Frequency estimation in data streams is one of the classical problems in streaming algorithms. Following much research, there are now almost matching upper and lower bounds for the trade-off needed between the number of samples and the space complexity of the algorithm, when the data streams are adversarial. However, in the case where the data stream is given in a random order, or is stochastic, only weaker lower bounds exist. In this work we close this gap, up to logarithmic factors. 

In order to do so we consider the needle problem, which is a natural hard problem for frequency estimation studied in (Andoni et al. 2008, Crouch et al. 2016). Here, the goal is to distinguish between two distributions over data streams with $t$ samples. The first is uniform over a large enough domain. The second is a planted model; a secret ''needle'' is uniformly chosen, and then each element in the stream equals the needle with probability $p$, and otherwise is uniformly chosen from the domain. It is simple to design streaming algorithms that distinguish the distributions using space $s \approx 1/(p^2 t)$. It was unclear if this is tight, as the existing lower bounds are weaker. We close this gap and show that the trade-off is near optimal, up to a logarithmic factor. 

Our proof builds and extends classical connections between streaming algorithms and communication complexity, concretely multi-party unique set-disjointness. We introduce two new ingredients that allow us to prove sharp bounds. The first is a lower bound for an asymmetric version of multi-party unique set-disjointness, where players receive input sets of different sizes, and where the communication of each player is normalized relative to their input length. The second is a combinatorial technique that allows to sample needles in the planted model by first sampling intervals, and then sampling a uniform needle in each interval. 
\end{abstract}

\section{Introduction}
\label{sec:intro}

The \emph{needle problem} is a basic question studied in the context of streaming algorithms for stochastic streams \cite{alon1999space, andoni2008better,guha2009revisiting,crouch2016stochastic,braverman2018revisiting}. The goal is to distinguish, using a space-efficient single-pass streaming algorithm, between streams sampled from two possible underlying distributions.

Setting notations, we let $t$ denote the number of samples, $s$ the space of the streaming algorithm, $n$ the domain size, and $p \in (0,1)$ the needle probability. The two underlying distributions are:
\begin{itemize}
    \item \textbf{Uniform:} sample $t$ uniform elements from $[n]$.
    \item \textbf{Planted:}  Let $x \in [n]$ be uniformly chosen (the ``needle''). Sample $t$ elements, where each one independently with probability $p$ equals $x$, and otherwise is sampled uniformly from $[n]$.
\end{itemize}
We will assume that $n=\Omega(t^2)$ so that with high probability, all elements in the stream (except for the needle in the planted model) are unique. The question is what space is needed to distinguish between the two models with high probability.

\paragraph{Sample-space tradeoffs for the needle problem.}
We start with describing some basic streaming algorithms for the needle problem, in order to build intuition. First, note that we need $p = \Omega(1/t)$ as otherwise the two distributions are statistically close, because with high probability the needle never appears in the planted model.

One possible algorithm is to check if there are two adjacent equal elements in the stream. This requires $t=\Theta(1/p^2)$ samples and space $s=\Theta(\log n)$. Another possible algorithm is to store the entire stream in memory, and check for a repeated element. This algorithm requires less samples, $t=\Theta(1/p)$, but more space, $s=t \log n$. Note that in both cases, we get a sample-space tradeoff of $st = \Theta((\log n)/p^2) $. One can interpolate between these two basic algorithms, but the value of the product $st$ remains the same in all of them.
This motivated the following conjecture, given explicitly in \cite{crouch2016stochastic} and implicitly in \cite{andoni2008better}.

\begin{conjecture}[Sample-space tradeoff for the needle problem]
\label{conj:main}
Any single-pass streaming algorithm which can distinguish with high probability between the uniform and planted models, where $p$ is the needle probability, $t$ the number of samples and $s$ the space, satisfies $p^2 st = \Omega(1)$.
\end{conjecture}

The best result to date towards \Cref{conj:main} is by Andoni {et al.} \cite{andoni2008better} who showed that $p^{2.5} s t^{1.5} = \Omega(1)$ (this bound is indeed weaker since $p=\Omega(1/t)$). Guha {et al.} \cite{guha2009revisiting} claimed to prove \Cref{conj:main} but later a bug was discovered in the proof, as discussed in \cite{crouch2016stochastic}. In this paper we establish \Cref{conj:main} up to logarithmic factors. We can also handle streaming algorithms which pass over the data stream multiple times, scaling linearly in the number of passes.

\begin{theorem}[Main theorem]
\label{thm:intro:needle}
Any $\ell$-pass streaming algorithm which can distinguish with high probability between the uniform and planted models, where $p$ is the needle probability, $t$ the number of samples and $s$ the space, satisfies $\ell p^2 st \log(t)= \Omega(1)$.
\end{theorem}

\subsection{Application: lower bound for frequency estimation in stochastic streams}

For many streaming problems, the current state-of-the-art streaming algorithms space requirements are known to be tight (up to poly-logarithmic terms) in the adversarial model, where the streams arrive in an adversarial order. 
Following a sequence of works on the random-order model \cite{munro1980selection,demaine2002frequency,guha2007space,chakrabarti2008robust,chakrabarti2008tight,andoni2008better,guha2009stream},
Crouch {et al.} \cite{crouch2016stochastic} initiated the study of \emph{stochastic streams}, where the streams are sampled from some underlying distribution. The question is if in this model one can attain better streaming algorithms compared to the adversarial model, utilizing the stochastic nature of the streams; or whether the existing lower bounds can be strengthened to this model as well. The needle problem we described is an example of a problem in the stochastic model.

A basic problem in the streaming literature, starting with the pioneering work of \cite{alon1999space}, is that of estimating the \emph{frequency moments} of a stream. Given a stream $x_1,\ldots,x_t$ of elements from $[n]$, let $f_x$ denote the number of times an element $x$ appears in the stream. The $k$-th frequency moment of the stream is
$$
F_k = \sum_{x \in [n]} f_x^k.
$$
In the adversarial model, there are matching upper and lower bounds of $\tilde{\Theta}(n^{1-2/k})$ \footnote{We use $\tilde{\Theta},\tilde{\Omega}$ to ignore poly-logarithmic terms.} on the space needed for a streaming algorithm to approximate $F_k$ \cite{chakrabarti2003near,indyk2005optimal}. It was conjectured by \cite{crouch2016stochastic} that the same lower bound also holds in the stochastic model. They showed that the result of \cite{andoni2008better} gives a somewhat weaker lower bound of $\tilde{\Omega}(n^{1-2.5/k})$ space, and that \Cref{conj:main}, if true, implies the tight bound of $\tilde{\Omega}(n^{1-2/k})$. \Cref{thm:intro:needle} thus verifies their conjecture, up to logarithmic terms, which still implies a lower bound of $\tilde{\Omega}(n^{1-2/k})$. We refer to \cite{crouch2016stochastic} for further details.

We note another related application, communicated to us by David Woodruff. McGregor {et al.} \cite{mcgregor2012space} studied streaming algorithms based on sub-sampling a data stream. In particular, one of the problems they studied is that of frequency estimation. They designed space-efficient streaming algorithms based on sub-sampling, and also gave matching lower bounds, based on the results of Guha {et al.} \cite{guha2009revisiting}. However, as later a bug was found in this latter work, the journal version of McGregor {et al.} \cite{mcgregor2016space} removed the lower bounds. Using \Cref{thm:intro:needle} the claimed lower bounds hold, up to a logarithmic factor.

\subsection{Proof approach}
\label{sec:intro:proof-approach}

We prove \Cref{thm:intro:needle} by a reduction to the unique set-disjointness problem in communication complexity. This is a common technique used to prove lower bounds for streaming algorithms \cite{chakrabarti2003near,bar2004information,andoni2008better,guha2009revisiting,braverman2018revisiting,kamath2021simple}. 

The basic idea is to partition the stream samples into intervals $I_1,\ldots,I_k$ and consider the stream distribution where we place a single needle uniformly in each interval, and sample the other elements in the stream uniformly. It is straightforward to show that any streaming algorithm which can distinguish this distribution from the uniform distribution using space $s$, can be used to construct a communication protocol that solves the $k$-party unique set-disjointness problem, where player $i$ gets a set of size $|I_i|$, and where each player sends $s$ bits. If for example we take the intervals to be of equal size $|I_1|=\ldots=|I_k|=t/k$, then using existing tight lower bounds for multi-party unique set-disjointness, one can prove tight sample-space lower bounds in the adversarial model\footnote{Concretely, the total communication of the protocol is $ks$, whereas the lower bound for $k$-party unique set-disjointness is $\Omega(t/k)$. Thus $ks =\Omega(t/k)$. Taking $k=pt$ gives $p^2 s t=\Omega(1)$.}. This was the approach taken by many of the previous works in this area \cite{chakrabarti2003near,bar2004information,andoni2008better,guha2009revisiting,braverman2018revisiting,kamath2021simple}. Our plan is to extend this approach to the stochastic model. However, this presents two new challenges. 

First, a simple calculation shows that the number of needles is $k \approx pt$ with high probability, but the gaps between needles are not uniform; for example, the two closest needles have a gap of $\approx p^2 t$. This necessitates taking intervals of very different lengths, if we still plan to place one needle per interval. In turn, this requires proving lower bounds on multi-party unique set-disjointness when the players receive inputs of different lengths. In this model, it no longer makes sense to measure the total communication of the protocols. Instead, we develop a new measure, which normalizes the communication of each player relative to their input length. We expand on this in \Cref{sec:intro:udisj}.

The second challenge is that using a single partition of the stream by intervals, and then planting a uniform needle in each interval, cannot induce the planted needle distribution. Instead, we need to carefully construct a distribution over sets of intervals, such that if then one places a uniform needle in each interval, the resulting stream distribution mimics exactly the planted distribution. We expand on this in \Cref{sec:intro:interval}.

\subsection{Multi-party unique set-disjointness with different set sizes}
\label{sec:intro:udisj}

We start by defining the standard multi-party unique set-disjointness problem. Let $k \ge 2$ denote the number of players. The players inputs are sets $S_1,\ldots,S_k \subset [n]$. They are promised that one of two cases hold:
\begin{itemize}
    \item \textbf{Disjoint}: the sets $S_1,\ldots,S_k$ are pairwise disjoint.
    \item \textbf{Unique intersection}: there is a common element $x \in S_1 \cap \ldots \cap S_k$, and the sets $S_1 \setminus \{x\},\ldots,S_k \setminus \{x\}$ are pairwise disjoint.
\end{itemize}
Their goal is to distinguish which case is it, while minimizing the communication\footnote{Formally, we consider randomized multi-party protocols in the \emph{blackboard model}, where at each turn one of the players writes a message on a common blackboard seen by all the players.}.

Observe that under any of the two promise cases, one of the players' inputs has size $|S_i| \le n/k+1$. A simple protocol is that such a player sends their input, which allows the other players to solve the problem on their own. This simple protocol sends $O(n/k \cdot \log n)$ bits. This can be further improved to $O(n/k)$ bits using the techniques of \cite{haastad2007randomized}.
A line of research 
\cite{alon1999space, bar2004information, chakrabarti2003near, gronemeier2009asymptotically,jayram2009hellinger, yang2022lifting} studied lower bounds. A tight lower bound was first achieved by \cite{gronemeier2009asymptotically}.

\begin{theorem}[\cite{gronemeier2009asymptotically,jayram2009hellinger}]
\label{thm:intro:udist-classic}
Any randomized communication protocol which solves the $k$-party unique set-disjointness problem must send $\Omega(n/k)$ bits.
\end{theorem}

As discussed in \Cref{sec:intro:proof-approach}, we need a fine-grained variant of the unique set-disjointness problem, where the set sizes are fixed and can be different between the players.

\begin{definition}[Fixed-size multi-party unique set-disjointness]
Let $s_1,\ldots,s_k \ge 1$. The $[s_1,\ldots,s_k]$-size $k$-party unique set-disjointness problem is a restriction of the $k$-party unique set-disjointness problem to input sets of size $|S_i|=s_i$.
\end{definition}

Consider protocols for the $[s_1,\ldots,s_k]$-size $k$-party unique set-disjointness problem. For any $i \in [k]$, one option is that the $i$-th player sends their input to the rest of the players, which requires sending $c_i=\Omega(s_i)$ bits. If the input sizes $s_1,\ldots,s_k$ are very different, it no longer makes sense to consider the total amount of bits sent by the players. Instead, we should normalize the number of bits sent by the $i$-th player $c_i$ by its input length $s_i$. We prove that with this normalization, the simple protocols are indeed optimal.

Towards this, we make the following definition: a $k$-party protocol $\Pi$ is called $[c_1,\ldots,c_k]$-bounded if in any transcript of $\Pi$, the $i$-th player sends at most $c_i$ bits.

\begin{theorem}[Lower bound for fixed-size multi-party unique set-disjointness]
\label{thm:intro:udisj}
Let $\Pi$ be a randomized $k$-party $[c_1,\ldots,c_k]$-bounded protocol, which solves with high probability the $[s_1,\ldots,s_k]$-size $k$-party unique set-disjointness problem, where $\sum s_i \le n/2$. Then
$$
\sum_{i \in [k]} \frac{c_i}{s_i} = \Omega(1).
$$
\end{theorem}

We conclude this subsection with three comments. First, the condition $\sum s_i \le n/2$ is a technical condition emerging from the proof technique; it suffices for our application, and we believe that it can be removed in future work. 

Next, it is known that the hard case for the standard  multi-party unique set-disjointness problem is when all the sets have about the same size, namely when $s_1=\ldots=s_k=\Theta(n/k)$. In this case \Cref{thm:intro:udisj} implies $\sum c_i = \Omega(n/k)$ which recovers \Cref{thm:intro:udist-classic}.

Last, we prove \Cref{thm:intro:udisj} by constructing a hard distribution over inputs, and then proving a lower bound for deterministic protocols under this distribution. The hard distribution is a natural one, the uniform distribution over inputs of sizes $s_1,\ldots,s_k$. For details see \Cref{thm:udisj_size}. Moreover, we show (\Cref{claim:udisj_equiv}) that \Cref{thm:intro:udisj} and \Cref{thm:udisj_size} are in fact equivalent.

\subsection{Efficient reduction of the needle problem to multi-party unique set-disjointness}
\label{sec:intro:interval}

We establish \Cref{thm:intro:needle} by reducing lower bounds for the needle problem to lower bounds for the unique set-disjointness, and then applying \Cref{thm:intro:udisj} (\Cref{thm:udisj_size} more precisely). To do so, we need a way of mapping inputs to the unique set-disjointness problem to inputs for a streaming algorithm. A natural way to do so, taken for example by \cite{andoni2008better}, is to partition the stream into intervals and assign one to each player. We follow the same approach but generalize it, so we can use it to simulate the planted distribution of the needle problem by random inputs to the unique set-disjointness problem.

Recall that $n$ denotes the domain size, $t$ the number of samples and $p$ the needle probability. Our goal will be to simulate the planted distribution using inputs to multi-party unique set-disjointness. In order to do so, we define \emph{interval systems}.

\begin{definition}[Interval systems]
An interval system $F$ is a family of pairwise disjoint non-empty intervals $F=\{I_1,\ldots,I_k\}$ with $I_1,\ldots,I_k \subset [t]$.
\end{definition}

Given an interval system $F$, we define a planted distribution $\planted[F]$ over streams $X \in [n]^t$ as follows:
\begin{enumerate}
    \item Sample uniform needle $x \in [n]$;
    \item In each interval $I \in F$ sample uniform index $i \in I$ and set $X_i=x$;
    \item Sample all other stream elements uniformly from $[n]$.
\end{enumerate}

Using \Cref{thm:intro:udisj}, we prove a space lower bound for streaming algorithms that can distinguish between the uniform distribution and the planted distribution for $F$. Here is where we exploit the fact that we can prove lower bounds for unique set-disjointness also when the set sizes vary between the players. We use the following notation: given an interval system $F$, its value is $\val(F)=\sum_{I \in F} \frac{1}{|I|}$.

\begin{lemma}
\label{lemma:intro:planted-interval-lb}
Let $F$ be an interval system.
Any streaming algorithm which with high probability distinguishes between $\planted[F]$ and the uniform distribution must use space 
$$
s = \Omega \left(\frac{1}{\val(F)} \right).
$$
\end{lemma}

In order to complete the reduction, we need to simulate the planted distribution using planted distributions for interval systems $F$. Clearly, this cannot be done using a single interval system, and hence we need to consider \emph{randomized} interval systems.

A randomized interval system $\cF$ is a distribution over interval systems $F$. The planted distribution $\planted[\cF]$ for $\cF$ is defined by first sampling $F \sim \cF$ and then $X \sim \planted[F]$. The value of $\cF$ is $\val(\cF)=\E_{F \sim \cF}[\val(F)]$. We can extend \Cref{lemma:intro:planted-interval-lb} to randomized interval systems.

\begin{lemma}
\label{lemma:intro:planted-rand-interval-lb}
Let $\cF$ be a randomized interval system.
Any streaming algorithm which with high probability distinguishes between $\planted[\cF]$ and the uniform distribution must use space 
$$
s = \Omega \left(\frac{1}{\val(\cF)} \right).
$$
\end{lemma}

To prove the lower bound for the needle problem, we need $\planted[\cF]$ to simulate exactly the planted distribution; we call such randomized interval systems \emph{perfect}.

\begin{definition}[Perfect randomized interval systems]
A randomized interval system $\cF$ is called \emph{perfect} if $\planted[\cF]$ is distributed exactly as the planted distribution.
\end{definition}

In light of \Cref{lemma:intro:planted-rand-interval-lb}, we need a perfect randomized interval system $\cF$ with as low a value as possible. 
It is relatively simple to show that if $\cF$ is perfect then $\val(\cF) = \Omega(p^2 t)$. The following theorem gives a construction nearly matching the lower bound.

\begin{theorem}
\label{thm:intro:interval}
There exists a perfect randomized interval system $\cF$ with
$\val(\cF) = O\left( p^2 t \log(t) \right)$.
\end{theorem}

\Cref{thm:intro:needle} now follows directly by combining \Cref{lemma:intro:planted-rand-interval-lb} and \Cref{thm:intro:interval}. 

\subsection{Related works}

In a seminal work, Miltersen et al. \cite{miltersen1995data} first observed connections between asymmetric communication complexity and its applications to data structures in the cell probe model. Since then, several works \cite{barkol2000tighter, jayram2003cell, puatracscu2006time, ba2010lower, chattopadhyay2018simulation} proved data structure lower bounds and streaming lower bounds via connections to asymmetric communication complexity lower bounds. To the best of our knowledge, all these works built on two-party communication problems. In contrast, we consider multi-party communication complexity in this work. It is interesting to ask if multi-party communication can provide more applications to data structure and streaming lower bounds.

Other than connections to data structure lower bounds and streaming lower bounds, Dinur {et al.} \cite{dinur2016memory} studied the needle problem in cryptography. It would be interesting to explore more connections between our work and cryptography.

\paragraph{Acknowledgements.} We thank David Woodruff for helpful discussions about streaming algorithms, and for insightful comments on an earlier version of this paper.

\paragraph{Paper organization.} We prove lower bounds for multi-party unique set-disjointness with fixed set sizes (\Cref{thm:intro:udisj}) in \Cref{sec:udisj}. We design an efficient reduction using interval systems (\Cref{lemma:intro:planted-interval-lb,lemma:intro:planted-rand-interval-lb}) in \Cref{sec:interval}. We combine both to prove our lower bound for the needle problem (\Cref{thm:intro:needle}) in \Cref{sec:needle}. We discuss open problems in \Cref{sec:open}.

\section{Lower bounds for asymmetric unique set-disjointness}
\label{sec:udisj}

We prove \Cref{thm:intro:udisj} in this section. First, we recall some definitions and fix some notations.

\paragraph{Notations.}
it will be convenient to identify sets with their indicator vectors; thus, we identify $X \in \bits^n$ with the set $\{i: X_i=1\} \subset [n]$. Let $k \ge 2$ denote the number of players. The players inputs are $X=(X_1,\ldots,X_k)$, where $X_i=(X_i(1),\ldots,X_i(n)) \in \bits^n$. It will be convenient to also define $X^j=(X_1(j),\ldots,X_k(j)) \in \bits^k$, the $j$-th coordinate for all the players for $j \in [n]$. In this section use boldface to denote random variables (such as $\rX, \rW$) to help distinguish them from non-random variables.

\paragraph{Protocols.} Let $\Pi$ be a protocol. Given an input $X$, we denote by $\Pi(X)$ the transcript of running $\Pi$ on $X$. We assume that every transcript also has an output value which is a bit determined by the transcript (for example, the last bit sent). A protocol solves a decision problem under input distribution $\nu$ with error $\delta$, if it outputs the correct answer with probability at least $1-\delta$ when the inputs are sampled from $\nu$. We will prove lower bounds on protocols that solve unique set-disjointness under a number of input distributions. As such, we may assume unless otherwise specified that the protocols are deterministic.

Finally, recall that we call $k$-party protocol $\Pi$ is called $[c_1,\ldots,c_k]$-bounded if in any transcript of $\Pi$, the $i$-th player sends at most $c_i$ bits.

\paragraph{multi-party unique set-disjointness.}
The $k$-party unique set-disjointness problem is defined on inputs coming from two promise sets:
\begin{itemize}
    \item \textbf{Disjoint}: $\F{0} = \{X \in (\bits^n)^k: \forall j \in [n], |X^j| \le 1\}$,
    \item \textbf{Unique intersection}: $\F{1} = \{X \in (\bits^n)^k: \exists j \in [n], |X^j|=k, \forall j' \ne j, |X^{j'}| \le 1\}$.
\end{itemize}

Towards proving \Cref{thm:intro:udisj}, our first step is to consider unique set-disjointness under product distribution which assign weight asymmetrically between the players.

\subsection{Lower bounds for product asymmetric distributions}

Let $\nu$ be a distribution over $[k]$. We denote by $\nu^n$ the distribution over $\rW \in [k]^n$, where we sample $\rW_j \sim \nu$ independently for all $j \in [n]$. 
We define two distributions, $\disprob{0}{\nu}$ supported on $\F{0}$ and $\disprob{1}{\nu}$ supported on $\F{1}$. 

\begin{definition}[Disjoint asymmetric distribution]
\label{def:prob0}
Let $\rX \in (\bits^n)^k$ be sampled as follows:
\begin{enumerate}
    \item Sample $\rW \sim \nu^n$.
    \item For each $j \in [n]$, if $\rW_j=i$ then we sample $\rX_i(j) \in \bits$ uniformly, and set $\rX_{i'}(j)=0$ for all $i' \ne i$.
\end{enumerate}
We denote by $\disprob{0}{\nu}$ the marginal distribution of $\rX$, and note that it is supported on $\F{0}$.
\end{definition}

\begin{definition}[Unique intersection asymmetric distribution]
\label{def:prob1}
Let $\rY \in (\bits^n)^k$ be sampled as follows:
\begin{enumerate}
    \item Sample $\rX \sim \disprob{0}{\nu}$.
    \item Sample $\rj \in [n]$ uniformly.
    \item If $\rj=j$ then we set $\rY^j=1^k$ and $\rY^{j'}=\rX^{j'}$ for all $j' \ne j$.
\end{enumerate}
We denote by $\disprob{1}{\nu}$ the marginal distribution of $\rY$, and note that it is supported on $\F{1}$.
\end{definition}

We denote by $\disprob{}{\nu}$ the mixture distribution, where we sample $\rb \in \bits$ uniformly, and then sample $\rX \sim \disprob{\rb}{\nu}$. Our main technical result is a communication lower bound on protocols which solve unique set-disjointness under input distribution $\disprob{}{\nu}$. We will later reduce the fixed set size case to this model.

\begin{theorem}
\label{thm:udisj_prob}
Fix $n, k \ge 1$. Let $\nu$ be a distribution on $[k]$.
Let $\Pi$ be a $[c_1,\ldots,c_k]$-bounded $k$-party deterministic protocol which solves the unique set-disjointness problem under input distribution $\disprob{}{\nu}$ with error $2\%$. Then
$$
\sum_{i \in [k]} \frac{c_i}{\nu(i)} = \Omega(n).
$$
\end{theorem}

We note that \Cref{thm:udisj_prob} is a generalization of the lower bound for symmetric case \cite{gronemeier2009asymptotically,jayram2009hellinger}, where $\nu(i)=1/k$ for all $i\in[k]$. In this case \Cref{thm:udisj_prob} gives that $\sum_{i} c_i = \Omega(n/k)$.

\subsubsection{Information theory framework}
We will use information theory to prove \Cref{thm:udisj_prob}. Although we assume that $\Pi$ has small error with respect to both $\disprob{0}{\nu}$ and $\disprob{1}{\nu}$, we will only study its information complexity with respect to $\disprob{0}{\nu}$. Below we let $\rW \in [k]^n, \rX \in (\bits)^n$ be jointly samples as in \Cref{def:prob0}. The following observation will play an important role.
    
\begin{observation}
\label{obs:independent}
Conditioned on $\rW=W$, the random variables $(\rX_i(j): i \in [k], j \in [n])$ are independent.
\end{observation}

We start by giving a general bound for individual communication based on information theory, which assumes only the existence of such $\rW$ under which $\rX_1,\ldots,\rX_k$ are independent.

\begin{lemma}
\label{lemma:asys_info_bits}
Let $\Pi$ be a $k$-party protocol which is $[c_1,\ldots,c_k]$-bounded.
Assume joint random variables $(\rW, \rX)$, where $\rX=(\rX_1,\ldots,\rX_k)$ are the players inputs, and such that for every value $W$ for $\rW$, the random variables $\rX_1|\rW=W,\ldots,\rX_k|\rW=W$ are independent. Then for each $i \in [k]$ we have
$$
c_i \ge I(\rX_i: \Pi(\rX) |\rW).
$$
\end{lemma}

\begin{proof}
We first set up some notations. We denote by $\pi$ a possible transcript for $\Pi$, and let $\pi_{<t}=(\pi_1,\ldots,\pi_{t-1})$ be a partial transcript. We let $\rpi=\Pi(\rX)$ denote the transcript when the protocol is run on $\rX$.

Fix a time step $t$ in the protocol, and a partial transcript $\pi_{<t}$. The next player to speak is determined by the transcript so far, so denote it by $\next(\pi_{<t})\in [k]$. We also denote by $\locs(\pi,i)=\{t:\next(\pi_{<t})=i\}$ the locations in transcript $\pi$ where player $i$ sent a bit. By our assumption $|\locs(\pi,i)|\le c_i$ for any transcript $\pi$.

Consider any value $W$ for $\rW$. Observe that conditioned on $\rpi_{<t}=\pi_{<t}$, the next bit sent $\rpi_t$ is a function of $\rX_i$ for $i=\next(\pi_{<t})$. If $i' \ne i$ then since $\rX_i|\rW=W, \rX_{i'}|\rW=W$ are independent we have
$$
I(\rX_{i'}:\rpi_t|\rW=W, \rpi_{<t}=\pi_{<t}) = 0.
$$
Since $\rpi_t \in \bits$, we can also trivially bound
$$
I(\rX_i:\rpi_t|\rW=W, \rpi_{<t}=\pi_{<t}) \le 1.
$$
Averaging over $\pi_{<t}$ and $W$ gives
$$
I(\rX_i:\rpi_t|\rW, \rpi_{<t}) \le \Pr[\next(\rpi_{<t})=i].
$$
Summing over $t$ then gives the result:
$$
I(\rX_i:\rpi|\rW) = \sum_t I(\rX_i:\rpi_t|\rW, \rpi_{<t}) = \E |\locs(\rpi, i)| \le c_i.
$$
\end{proof}

We shorthand $\rpi=\Pi(\rX)$ below.
Using \Cref{lemma:asys_info_bits},  \Cref{obs:independent} and the data processing inequality\footnote{If $\rx,\ry,\rz$ are random variables, where $\rx,\ry$ are independent, then  $I(\rx \ry:\rz) \ge I(\rx:\rz)+I(\ry:\rz)$.} give
$$
c_i \ge I(\rX_i:\rpi|\rW) \ge \sum_{j \in [n]} I(\rX_i(j): \rpi|\rW).
$$
Towards proving \Cref{thm:udisj_prob}, consider the expression
$$
\sum_{i \in [k]} \frac{c_i}{\nu(i)} \ge
\sum_{i \in [k]} \frac{1}{\nu(i)} I(\rX_i:\rpi|\rW) \ge
\sum_{i \in [k]} \frac{1}{\nu(i)} \sum_{j \in [n]} I(\rX_i(j): \rpi|\rW)
$$
We define below
$$
L := \frac{1}{n} \sum_{i \in [k]} \frac{1}{\nu(i)} \sum_{j \in [n]} I(\rX_i(j): \rpi|\rW)
$$
The following lemma thus proves \Cref{thm:udisj_prob}.
\begin{lemma}
\label{lemma:info_mixed_lb}
$L = \Omega(1)$.
\end{lemma}
We prove \Cref{lemma:info_mixed_lb} in the next subsection, via a reduction to protocols for the $k$-bit AND function.

\subsubsection{Reduction to the information complexity of the AND function}

In this section, we consider the $k$-bit AND function and its information complexity. Let $\Lambda$ be a $k$-party protocol for it: each of the $k$ players receive as input a bit, and their goal is to compute their AND. Namely, to check if they are all equal to $1$.

Let $\rb \in \bits$ be a random bit. For $i \in [k]$, let $e_i[\rb] \in \{0,1\}^k$ denote the vector with $\rb$ at coordinate $i$ and $0$ everywhere else. The following lemma reduces proving \Cref{lemma:info_mixed_lb} to analyzing the information of protocols for $k$-bit AND which make small error on only two inputs: the all-zero and all-one inputs.

\begin{lemma}
\label{lemma:pi_to_lambda}
There is a public-randomness $k$-party protocol $\Lambda$ for the $k$-bit AND function, using public-randomness $\rR$, with the following guarantees:
\begin{enumerate}
    \item $\Lambda$ has error at most $8\%$ with respect to the inputs $0^k$ and $1^k$.
    \item $L = \sum_{i \in [k]} I(\rb, \Lambda(e_i[\rb], \rR) | \rR)$.
\end{enumerate}
\end{lemma}

We prove \Cref{lemma:pi_to_lambda} in the remainder of this subsection. First, let $\rd \in [k], \rU \in \bits^k$ be jointly sampled as follows:
\begin{enumerate}
    \item Sample $\rd \in [k]$ according to $\nu$.
    \item Given $\rd=d$, sample $\rU_d \in \bits$ uniformly and set $\rU_i=0$ for all $i \ne d$.
\end{enumerate}
Let $\sigma=\sigma(\nu)$ denote the marginal distribution of $\rU$, and observe that it is the same as that of $\rX^j$ for any $j \in [n]$. In fact, the joint distribution of $(\rd, \rU)$ is the same as $(\rW_j, \rX^j)$ for any $j$. The next claim uses this to extract a protocol $\Lambda$ for $k$-bit AND from $\Pi$, such that it has related information complexity measures, and a small error with respect to the inputs $0^k$ and $1^k$.

\begin{claim}
\label{claim:pi_to_lambda1}
There is a (public randomness) $k$-party protocol $\Lambda$ for the $k$-bit AND function, using public randomness $\rR$, with the following properties:
\begin{enumerate}
    \item $\Lambda$ has error at most $8\%$ with respect to the inputs $0^k$ and $1^k$.
    \item $I(\rU_i:\Lambda(\rU, \rR)|\rd,\rR) =  \frac{1}{n} \sum_{j=1}^n I(\rX_i(j):\rpi|\rW)$ for all $i \in [k]$.
\end{enumerate}
\end{claim}

\begin{proof}
We first define the protocol $\Lambda$. Let $U \in \bits^k$ denote the input for the AND function. First, using public randomness, sample $\rj \in [n]$ uniformly; then sample $\rW_{-\rj}=(\rW_{j'}: j' \ne \rj) \sim \nu^{n-1}$. Conditioned on $\rj=j, \rW_{-j}=W_{-j}$, the $i$-th player then constructs their input $\rX_i$ for $\Pi$ as follows: set $\rX_i(j)=U_i$ and sample $\rX_i(j')|\rW_{j'}=W_{j'}$ using private randomness. The players then run the protocol $\Pi$ on their joint inputs $\rX=(\rX_1,\ldots,\rX_k)$. Note that the public randomness used is $\rR=(\rj,\rW_{-\rj})$.

To prove the first claim, observe that if the input to the AND function $\rU$ is distributed as $\rU \sim \sigma$, then $\rX \sim \disprob{0}{\nu}$; and if $U=1^k$ then $\rX \sim \disprob{1}{\nu}$.
Thus $\Lambda$ has error at most $2\%$ with respect to the uniform mixture of the input distributions $\sigma$ and $1^k$. Thus with respect to the input $1^k$, the error is at most $4\%$. Since $\sigma(0^k)=1/2$, the error with respect to the input $0^k$ is at most $8\%$.

For the second claim, note that conditioned on $\rR=R=(j,W_{-j})$, the joint distribution of $(\rd, \rU, \Lambda(\rU, R))$ and of $(\rW_j,\rX^j, \pi)$ is identical. Thus
$$
I(\rU_i:\Lambda(\rU, R)|\rd, \rR=R) = I(\rX_i(j):\rpi|\rW_j, \rj=j, \rW_{-\rj}=W_{-j})
$$
Averaging over $R$ gives
\begin{align*}
I(\rU_i:\Lambda(\rU, R)|\rd, \rR) &= \frac{1}{n} \sum_{j \in [n]} I(\rX_i(j):\rpi|\rW_j, \rj=j, \rW_{-\rj}=W_{-j})\\
&= \frac{1}{n} \sum_{j \in [n]} I(\rX_i(j):\rpi|\rW).
\end{align*}
\end{proof}

\begin{proof}[Proof of \Cref{lemma:pi_to_lambda}]
Let $\Lambda$ be the protocol given by \Cref{claim:pi_to_lambda1}. Then
$$
L = \sum_{i \in [k]} \frac{1}{\nu(i)} I(\rU_i: \Lambda(\rU,\rR)|\rd, \rR).
$$
Simplifying the inner terms give
\begin{align*}
\frac{1}{\nu(i)} I(\rU_i:\Lambda(\rU,\rR)|\rd, \rR) &= \frac{1}{\nu(i)} \sum_{j \in [k]} \nu(j) \cdot I(\rU_i:\Lambda(\rU,\rR)|\rd=j, \rR) \\
&= I(\rU_i:\Lambda(\rU,\rR)|\rd=i, \rR)
\end{align*}
Note that conditioned on $\rd=i$, the joint distribution of $(\rU_i, \rU)$ is the same as $(\rb, e_i[\rb])$. Thus
$$
L = \sum_{i \in [k]} I(\rb:\Lambda(e_i[\rb], \rR)| \rR).
$$
\end{proof}

\subsubsection{Bounding the information complexity of AND functions}

We prove the following lemma in this subsection, which then proves \Cref{thm:udisj_prob} given \Cref{lemma:asys_info_bits}, \Cref{lemma:info_mixed_lb} and \Cref{lemma:pi_to_lambda}.

\begin{lemma}
\label{lemma:and_info_lb}
Let $\Lambda$ be a (public randomness) protocol for the $k$-bit AND function, using public randomness $\rR$, such that it has error at most $8\%$ with respect to the inputs $0^k$ and $1^k$. Then
$$
\sum_{i \in [k]} I(\rb, \Lambda(e_i[\rb], \rR) | \rR) = \Omega(1).
$$
\end{lemma}

\Cref{lemma:and_info_lb} is very similar to previous lower bounds in the literature on information complexity \cite{bar2004information,chakrabarti2003near,gronemeier2009asymptotically}.
We need the following setup. Sample jointly $\re \in [k], \rV \in \bits^k$ as follows:
\begin{enumerate}
    \item Sample $\re \in [k]$ uniformly.
    \item Given $\re=e$, sample $\rV_e \in \bits$ uniformly and set $\rV_i=0$ for all $i \ne e$.
\end{enumerate}
Given a protocol $\Lambda$ using public randomness $\rR$, its conditional information complexity is
$$
\text{CIC}(\Lambda) = I(\rV : \Lambda(\rV, \rR) | \re, \rR).
$$
This quantity comes up naturally in the study of unique disjointness using information complexity, which started with the seminal work of \cite{bar2004information}. Gronemeier \cite{gronemeier2009asymptotically} and Jayram \cite{jayram2009hellinger} proved a tight lower bound on this quantity.

\begin{theorem}
\label{thm:gronemeir}
[\cite{gronemeier2009asymptotically,jayram2009hellinger}]
$\text{CIC}(\Lambda) = \Omega(1/k)$.
\end{theorem}

In fact, the proof (although not explicitly stated as such) only relies on the assumption that $\Lambda$ has error $\le 30\%$ on both the all-zero and all-one inputs (for a full proof see Gronemeier's thesis \cite{gronemeier2010information}). As such, it applies to our protocol $\Lambda$. The following claim connects $\text{CIC}(\Lambda)$ to the quantity we aim to bound, and concludes the proof of \Cref{lemma:and_info_lb} and hence also of \Cref{thm:udisj_prob}.

\begin{claim}
$\sum_{i \in [k]} I(\rb : \Lambda(e_i(\rb), \rR) | \rR) = k \cdot \text{CIC}(\Lambda)$.
\end{claim}

\begin{proof}
\begin{align*}
k \cdot \text{CIC}(\Lambda) &= k \cdot I(\rV : \Lambda(\rV, \rR) | \re, \rR) \\
&= \sum_{i \in [k]} I(\rV : \Lambda(\rV, \rR) | \re=i, \rR)\\
&= \sum_{i \in [k]} I(\rb : \Lambda(e_i(\rb), \rR) | \rR).
\end{align*}
\end{proof}

\subsection{Extension to sub-distributions}

It will be convenient to extend \Cref{thm:udisj_prob} to sub-distributions. A sub-distribution $\nu$ on $[k]$ satisfies $\nu(i) \ge 0$ and $\sum \nu(i) \le 1$. We extend the definition of $\disprob{0}{\nu}$, $\disprob{1}{\nu}$ to sub-distributions as follows.

We first describe how to sample $\rX \sim \disprob{0}{\nu}$. For each $j \in [n]$, with probability $\nu(i)$ sample $\rX_i(j) \in \bits$ uniformly, and set $\rX_{i'}(j)=0$ for all $i \ne i'$; and with probability $1-\sum \nu(i)$ set $\rX_i(j)=0$ for all $i$. To sample $\rY \sim \disprob{1}{\nu}$ we follow the same process as for the distributional case: first sample $\rX \sim \disprob{0}{\nu}$, then sample a uniform $\rj \in [n]$ and set $\rY^{\rj}=1^k$ and $\rY^{j'}=\rX^{j'}$ for all $j' \ne j$. We denote by $\disprob{}{\nu}$ the even mixture of $\disprob{0}{\nu}$ and $\disprob{1}{\nu}$.
The following theorem extends \Cref{thm:udisj_prob} to sub-distributions.

\begin{theorem}
\label{thm:udisj_subprob}
Fix $n, k \ge 1$. Let $\nu$ be a sub-distribution on $[k]$.
Let $\Pi$ be a $[c_1,\ldots,c_k]$-bounded protocol which solves the distributional unique set-disjointness under input distribution $\disprob{}{\nu}$ with error $2\%$. Then
$$
\sum_{i \in [k]} \frac{c_i}{\nu(i)} = \Omega(n).
$$
\end{theorem}

\begin{proof}
Extend $\nu$ to a distribution $\nu'$ on $[k+1]$ by setting $\nu'(i)=\nu(i)$ for $i \in [k]$ and $\nu'(k+1)=1 - \sum \nu(i)$. Extend $\Pi$ to a protocol $\Pi'$ for $k+1$ players where player $k+1$ does not participate in the protocol at all. Thus $\Pi'$ is a $[c_1,\ldots,c_k,0]$-bounded protocol. The proof follows by applying \Cref{thm:udisj_prob} to $\Pi'$ and $\nu'$.
\end{proof}

\subsection{Extension for fixed set sizes}

We now use the results we proven to deduce  \Cref{thm:intro:udisj}. Namely, the lower bound for fixed set sizes. We first set some notations.

Let $\s=[s_1,\ldots,s_k]$ denote the set sizes where $s_i \ge 1$ and $\sum s_i \le n$. Define
$$
\Fsize{}{\s} = \{X \in (\bits^n)^k: \forall i \in [k], |X_i|=s_i\}.
$$
For $b \in \bits$ define $\Fsize{b}{\s}=\F{b} \cap \Fsize{}{\s}$ and $\dissize{b}{\s}$ to be the uniform distribution over $\Fsize{b}{\s}$. Our hard distribution $\dissize{}{\s}$ will be an even mixture between $\dissize{0}{\s}$ and $\dissize{1}{\s}$. Equivalently, sample $\rb \in \bits$ uniformly and take $\rX \sim \dissize{\rb}{\s}$.
We prove a communication lower bound on protocols which solve unique set-disjointness under input distribution $\dissize{}{\s}$. 

\begin{theorem}
\label{thm:udisj_size}
Let $\s=[s_1,\ldots,s_k]$ with $\sum s_i \le n/2$.
Let $\Pi$ be a $[c_1,\ldots,c_k]$-bounded $k$-party protocol which solves the unique set-disjointness problem under input distribution $\dissize{}{\s}$ with error $1\%$. Then 
$$
\sum_{i \in [k]} \frac{c_i}{s_i} = \Omega(1).
$$
\end{theorem}

It is clear that \Cref{thm:udisj_size} implies \Cref{thm:intro:udisj}, but in fact they are equivalent. Before proving it we need the following claim.

\begin{claim}
\label{claim:F_size_perm}
Let $b \in \{0,1\}$, $X \in \Fsize{b}{\s}$. Let $\rSigma$ be a random permutation of $[n]$ and let $\rSigma(X)$ denote the result of applying $\rSigma$ to $X$. Then $\rSigma(X)$ is uniform in $\Fsize{b}{\s}$.
\end{claim}

\begin{proof}
The claim follows as permutations on $[n]$ act transitively on $\Fsize{b}{\s}$. Namely, for any $X,X' \in \Fsize{b}{\s}$ there exists a permutation $\Sigma$ on $[n]$ such that $\Sigma(X)=X'$. This implies that a uniform permutation maps $X$ to a uniform element in the domain $\Fsize{b}{\s}$.
\end{proof}

\begin{claim}
\label{claim:udisj_equiv}
\Cref{thm:intro:udisj} and \Cref{thm:udisj_size} are equivalent.
\end{claim}

\begin{proof}
We are comparing the multi-party unique set-disjointess problem for sizes $\s=[s_1,\ldots,s_k]$ in two settings: worst-case inputs, and uniform inputs. Clearly, a protocol for worst-case inputs implies one under uniform inputs with the same communication and error guarantees. In the other direction, let $X \in \Fsize{b}{\s}$ be any input for unique set-disjointness. The players, using public randomness, sample a uniform permutation $\rSigma$ on $[n]$, and each applies it to their input. By \Cref{claim:F_size_perm} we know that $\rSigma(X)$ is distributed as $\dissize{b}{\s}$. They can now apply a protocol that solves unique-set disjointness under input $\dissize{}{\s}$. 
\end{proof}

We now turn to prove \Cref{thm:udisj_size}.

\begin{proof}[Proof of \Cref{thm:udisj_size}]
First, note that may assume $c_i \ge 1$ for all $i$, since we can remove players with $c_i=0$ from the game, as they are not allowed to send any bits. 

Let $\Pi$ be a protocol as assumed in \Cref{thm:udisj_size}. Namely, it is $[c_1,\ldots,c_k]$-bounded and solves unique set-disjointness under input distribution $\dissize{}{\s}$ with error $1\%$, where $\s=[s_1,\ldots,s_k]$ satisfies $\sum s_i \le n/2$. We will use it to design a $[c_1+1,\ldots,c_k+1]$-bounded protocol $\Pi'$ which solves unique set-disjointness in a specific sub-distributional case with error $2\%$, and then appeal to \Cref{thm:udisj_subprob}.

Next, define a sub-distribution $\nu$ on $[k]$ by $\nu(i) = \frac{s_i}{4n}$.
We consider its corresponding distributional input $\disprob{}{\nu}$ on inputs of size $n/2$ bits. Let $\rX=(\rX_1,\ldots,\rX_k) \sim \disprob{}{\nu}$ where $\rX \in (\bits^{n/2})^k$. Each $\rX_i$ is distributed Binomially $\text{Bin}(n/2, \nu(i))$ with expected size $\E[|\rX_i|] = \frac{s_i}{2}$. Thus by the Hoeffding bound,
$$
\Pr[|\rX_i| > s_i] \le \exp(-s_i/6).
$$
Let $E$ denote the event that $|\rX_i| > s_i$ for some $i \in [k]$. Then
$$
\Pr[E] \le \sum_{i \in [k]} \exp(-s_i/6).
$$

We first analyze the case that $\Pr[E] \ge 1\%$. In this case, since $c_i \ge 1$ by assumption, and since $\frac{1}{x} \ge C \exp(-x/6)$ for some absolute constant $C>0$ for all $x \ge 1$, we get
$$
\sum_{i \in [k]} \frac{c_i}{s_i} \ge C \sum_{i \in [k]} \exp(-s_i/6) \ge C \Pr[E] = \Omega(1). 
$$
From now on we assume $\Pr[E] < 1\%$.

We now design the protocol $\Pi'$. 
First, each player checks if their input $X_i$ satisfies $|X_i|>s_i$. If so, the protocol aborts. This requires each player to send one bit, and by assumption it aborts with probability at most $1\%$. Otherwise, each player extends their input $X_i$ to a new input $Y_i \in \bits^n$ of size $|Y_i|=s_i$ as follows.

Before the protocol starts, the players agree ahead of time on pairwise disjoint subsets $T_1,\ldots,T_k$ with $|T_i|=s_i$, supported in the last $n/2$ coordinates (so they do not overlap the inputs $X_1,\ldots,X_k$). Now, the $i$-th player adds arbitrary $s_i-|X_i|$ elements from $T_i$ to their set $X_i$; we denote the new input $Y_i \in \bits^n$. Note that $Y=(Y_1,\ldots,Y_k)$ satisfies the same promise as $X=(X_1,\ldots,X_k)$; namely, either they are pairwise disjoint, or they have a common element and except for it they are pairwise disjoint.

We would like to apply $\Pi$ to $Y$. However we cannot quite yet; while it is true that $Y \in \Fsize{0}{\s}$ or $Y \in \Fsize{1}{\s}$, its distribution is not uniform in the sets. However, here we can apply \Cref{claim:F_size_perm} to make the distribution of $Y$ uniform in the respective family. The players use public randomness to sample a permutation $\rSigma$ on $[n]$ and apply it to $Y$. Now we can apply $\Pi(\rSigma(Y))$ which would give the correct with error $2\%$ by assumption. The proof now follows from \Cref{thm:udisj_subprob}.
\end{proof}

\section{Interval systems}
\label{sec:interval}

Recall that our plan is to use the lower bounds for multi-party unique set-disjointness in order to prove lower bounds for streaming algorithms for the needle problem. In order to effectively embed the inputs for unique set-disjointness inside streams, we introduce a combinatorial construct that we call \emph{interval systems}. 

\begin{definition}[Interval]
An interval is a non-empty set of the form $I=\{a,a+1,\ldots,b\}$ for some $a \le b$.
\end{definition}

\begin{definition}[Interval systems]
A $[t]$-interval system is a set $F=\{I_1,\ldots,I_k\}$ of $k$ pairwise disjoint intervals supported in $[t]$. If we want to specify the number of intervals, we say $F$ is a $[t,k]$-interval system. 
\end{definition}

\begin{definition}[Randomized interval systems]
A randomized $[t]$-interval system $\cF$ is a distribution over $[t]$-interval systems $F$. Similarly, a randomized $[t,k]$-interval system $\cF$ is a distribution over $[t,k]$-interval systems $F$.
\end{definition}

Next, we define for an interval system a corresponding distribution over sets $T \subset [t]$. 

\begin{definition}[Set distribution for interval systems]
Let $F$ be a $[t]$-interval system. We denote by $\set(F)$ the distribution over sets $T \subset [t]$ obtained by choosing uniformly one element from each interval $I \in F$. 

If $\cF$ is a randomized $[t]$-interval system, then we define $\set(\cF)$ as follows: first sample $F \sim \cF$ and then sample $T \sim \set(F)$.
\end{definition}

Observe that if $\cF$ is a randomized $[t,k]$-interval system, then $\set(\cF)$
is a distribution over $k$-sets in $[t]$ (a $k$-set is a set of size $k$). Our goal will be to simulate the uniform distribution over $k$-sets in $[t]$. We call such randomized interval systems \emph{perfect}.

\begin{definition}[Perfect interval systems]
A randomized $[t,k]$-interval system $\cF$ is called \emph{perfect} if $\set(\cF)$ is the uniform distribution over all $k$-sets in $[t]$.
\end{definition}

There are many ways to construct perfect randomized $[t,k]$-interval systems. For example, a naive way is to sample $k$ uniform coordinates $i_1,\ldots,i_k \in [t]$, and then take the distribution over $F=\{\{i_1\},\ldots,\{i_k\}\}$. 
However, for an efficient reduction, we would need interval systems with as long intervals as possible. Technically, the efficiency of the reduction will be controlled by the following notion of \emph{value} of interval systems.

\begin{definition}[Value of interval systems]
Let $F$ be a $[t]$-interval system. Its value is
$$
\val(F) = \sum_{I \in F} \frac{1}{|I|}.
$$
If $\cF$ is a randomized $[t]$-interval system then its value is
$$
\val(\cF) = \E_{F \sim \cF} \left[ \val(F) \right].
$$
\end{definition}

In order to prove strong lower bounds on streaming algorithms, we would need a perfect distribution over $[t,k]$-intervals with as low a value as possible. The following claim gives a lower bound for this.

\begin{claim}
\label{claim:int_system_lb}
Let $F$ be a $[t,k]$-interval system. Then
$$
\val(F) \ge \frac{k^2}{t}.
$$
\end{claim}

\begin{proof}
Let $F=\{I_1,\ldots,I_k\}$ where $|I_i|=s_i$. We have $\sum s_i \le t$, and $\val(F)=\sum \frac{1}{s_i}$. This expression is minimized when all the $s_i$ are the equal, and hence
$$
\val(F) \ge k \cdot \frac{k}{ \sum s_i} \ge \frac{k^2}{t}.
$$
\end{proof}

Our main technical result in this section is a construction of a perfect randomized $[t,k]$-interval system with value close to optimal. We do so by designing a randomized algorithm that samples $[t,k]$-interval systems. We will show that its output distribution is perfect, and of value close to the minimum given by \Cref{claim:int_system_lb}.

It will be convenient to make the following definition of ``shifting'' an interval or an interval system. For an interval $I=[a,b]$ and an integer $c$, define $I+c=[a+c,b+c]$. For an interval system $F=\{I_1,\ldots,I_k\}$ define $F+c=\{I_1+c,\ldots,I_k+c\}$.

\begin{algorithm}[H]
\label{alg:interval}
\caption{SampleIntervalSystem}
\DontPrintSemicolon
\KwIn{$t \ge 1$, $k \ge 0$ with $k \le t$}
\KwOut{$[t,k]$-interval system $F$}
\BlankLine
\uIf{$k=0$}{
    \Return{$F=\{\}$}
}
\uElseIf{$k=1$}{
    \Return{$F=\{[t]\}$}
}
\Else{
    Let $s=\lceil t/2 \rceil$\\
    Sample $\rj \in \{0,\ldots,k\}$ with probability $\Pr[\rj=j] = \frac{{s \choose j} {t-s \choose k-j}}{{t \choose k}}$\\
    Compute $F_1=\text{SampleIntervalSystem}(s,\rj)$\\
    Compute $F_2=\text{SampleIntervalSystem}(t-s,k-\rj)$\\
    \Return{$F = F_1 \cup (F_2+s)$}
}
\end{algorithm}
\;\\
We denote by $\cF[t,k]$ the randomized $[t,k]$-interval system obtained by running $\text{SampleIntervalSystem}(t,k)$.

\begin{claim}
$\cF[t,k]$ is perfect.
\end{claim}

\begin{proof}
The proof is by induction on $k,t$. If $k=0$ or $k=1$ this is clear from the base cases of the algorithm. If $k \ge 2$, then we sample the number of elements $j$ in the interval $[s]$
with the same probability as a uniform $k$-set in $[t]$ would. By induction, the distribution $\cF[s,j]$ of $F_1$ is a perfect randomized $[s,j]$ interval system; and the distribution $\cF[t-s,k-j]$ of $F_2$ is a perfect randomized $[t-s,k-j]$ interval system. The claim follows.
\end{proof}

We next analyze the value of $\cF[t,k]$; to simplify the analysis, we restrict to the case $t$ is a power of two. This suffices for our application, and we expect the bound to extend to general $t$ with minimal modifications. We assume below that all logarithms are in base two. 

\begin{lemma}
\label{lemma:alg_value}
Assume $t$ is a power of two. Then $\val(\cF[t,k]) \le \frac{k^2 \log(2t)}{t}$.
\end{lemma}

In order to prove \Cref{lemma:alg_value}, we will need the following technical claim, computing first and second moments for the distribution over $\rj$ in the algorithm.

\begin{claim}
\label{claim:p_moments}
Let $t,k \ge 1$, $t$ even, and $0 \le j \le k$. Define $p(t,k,j)=\frac{{t/2 \choose j}{t/2 \choose k-j}}{{t \choose k}}$. Then
$$
\sum_{j=0}^k p(t,k,j) \cdot j = \frac{k}{2}
$$
and
$$
\sum_{j=0}^k p(t,k,j) \cdot j^2 \le \frac{k(k+1)}{4}.
$$
\end{claim}

\begin{proof}
Let $s=t/2$. Let $T$ be a uniform subset of $[t]$ of size $k$. Then $p(t,k,j)=\Pr[|T \cap [s]|=j]$. Hence
\begin{align*}
\sum_{j=0}^k p(t,k,j) \cdot j &= \E_{T} \left[ \sum_{i \in [s]} \one{i \in T}  \right]
= \sum_{i \in [s]} \Pr[i \in T] = s \cdot \frac{k}{2s} = \frac{k}{2}
\end{align*}
and
\begin{align*}
\sum_{j=0}^k p(t,k,j) \cdot j^2 &= \E_{T} \left[ \sum_{i,j \in [s]} \one{i \in T} \cdot \one{j \in T} \right]
= \sum_{i,j \in [s]} \Pr[i,j \in T] \\
&= s \cdot \frac{k}{2s} + s(s-1) \frac{k(k-1)}{2s(2s-1)} \le \frac{k}{2} + \frac{k(k-1)}{4} = \frac{k(k+1)}{4}.
\end{align*}
\end{proof}

\begin{proof}[Proof of \Cref{lemma:alg_value}]
Let $f(t,k) = t \cdot \val(\cF[t,k])$.
We have $f(t,0)=0, f(t,1)=1$ and $f(t,k)=0$ if $k>t$. The definition of $f(t,k)$ for $k \ge 2$ is recursive. Let $p(t,k,j)=\frac{{t/2 \choose j}{t/2 \choose k-j}}{{t \choose k}}$. Then
$$
\val(\cF[t,k]) = \sum_{j=0}^k p(t,k,j) \left( \val(\cF[t/2,j]) + \val(\cF[t/2,k-j]) \right).
$$
which implies
$$
f(t,k) = 4 \sum_{j=0}^k p(t,k,j) f(t/2, j).
$$
It will be instructive to compute $f(t,2)$:
$$
f(t,2) = \frac{t}{t-1} + \frac{t-2}{t-1} f(t/2,2) \le 2 + f(t/2,2) \le 2 \log(t).
$$
We will prove by induction that
$$
f(t,k) \le k^2 + k(k-1) \log(t).
$$
We already verified this for $k=0,1,2$. For $k \ge 3$ we have by induction:
$$
f(t,k)  \le 4 \sum_{j=0}^k p(t,k,j) \left(j^2 + j(j-1) \log(t/2)\right).
$$
Applying \Cref{claim:p_moments} gives
\begin{align*}
f(t,k) &\le k(k+1) + k(k-1) \log(t/2) \\
& = 2k + k(k-1) \log(t) \\
& \le k^2 + k(k-1) \log(t).
\end{align*}
Finally we get
$$
\val(\cF[t,k]) = \frac{f(t,k)}{t} \le \frac{k^2 + k(k-1) \log(t)}{t} \le \frac{k^2 \log(2t)}{t}.
$$
\end{proof}

Our application for streaming algorithms for the needle problem has an additional restriction, that the total length of the intervals in the interval system be bounded away from $t$. We refer to such interval systems as \emph{valid}.

\begin{definition}[Valid interval systems]
A $[t]$-interval system $F$ is called \emph{valid} if $\sum_{I \in F} |I| \le t/2$. A randomized $[t]$-interval system $\cF$ is called valid if all $[t]$-interval systems $F$ in its support are valid.
\end{definition}

We next show how to refine a an interval system to obtain a valid randomized interval system, while preserving the sets distribution, and without increasing the value too much. 

\begin{lemma}
\label{lemma:refine_valid}
Assume $k \le t/6$. Let $F$ be a $[t,k]$-interval system. Then there exists a randomized $[t,k]$-interval system $\cF$ such that:
\begin{enumerate}
    \item $\set(\cF)=\set(F)$
    \item $\val(\cF)\le 5 \cdot \val(F)$
    \item $\cF$ is valid
\end{enumerate}
\end{lemma}

\begin{proof}
Let $F=\{I_1,\ldots,I_k\}$. Given an interval $I_i$ define $\ell_i = \min(3, |I_i|)$. Partition $I_i$ into $\ell_i$ intervals $\{I_{i,a}: a \in [\ell_i]\}$ of as equal length as possible, and observe that
$$
\frac{|I_i|}{5} \le |I_{i,a}| \le \frac{|I_i|}{3} + 1  \quad \forall a \in [\ell_i].
$$
Let $p_{i,a}=\frac{|I_{i,a}|}{|I_i|}$. We define a randomized $[t,k]$-interval system $\cF$, where for each $i \in [k]$ independently, we replace $I_i$ with one of its sub-intervals. Concretely, we choose $a \in [\ell_i]$ with probability $p_{i,a}$ and replace $I_i$ with $I_{i,a}$.
We now prove the claims.
\begin{enumerate}
    \item Observe that sampling a uniform element $x \in I_i$ can equivalently be sampled by first sampling $a \in [\ell_i]$ with probability $p_{i,a}$, and then sampling a uniform element $x \in I_{i,a}$. This implies that $\set(\cF)=\set(F)$.
    
    \item Since $|I_{i,a}| \ge |I_i| / 5$ for all $i,a$, the claim holds for any $F'$ in the support of $\cF$, and hence also for $\cF$.

    \item Since $|I_{i,a}| \le (|I_i|+1) / 2$ for all $i,a$, we have for any $F'=\{I_{1,a_1},\ldots,I_{k,a_k}\}$ in the support of $\cF'$ that
    $$
    \sum_{i \in [k]} |I_{i,a_i}| \le k+\frac{1}{3} \sum_{i \in [k]} |I_i| \le k+\frac{t}{3} \le \frac{t}{2}
    $$
    where the last inequality follows since we assume $k \le t/6$.
\end{enumerate}
\end{proof}

\Cref{lemma:refine_valid} applies also to randomized $[t,k]$-interval systems, by applying it to any interval system in their support. The following lemma summarizes all the facts we would need by applying it to $\cF[t,k]$.

\begin{lemma}
\label{lemma:interval_summary}
Let $k,t \ge 1$. Assume $t$ is a power of two and $k \le t/6$. Then there exists a valid perfect randomized $[t,k]$-interval system $\cF$ with
$$
\val(\cF) \le \frac{10 k^2 \log(t)}{t}.
$$
\end{lemma}

\section{Lower bound for the needle problem}
\label{sec:needle}

We prove \Cref{thm:intro:needle} in this section, by combining our lower bound for unique set-disjointness with fixed set sizes (\Cref{thm:udisj_size}) with the efficient reduction given by interval systems (\Cref{lemma:interval_summary}).

First, we recall the parameters: $n$ denotes the size of the domain, $t$ the number of samples and $p$ the needle probability. We assume throughout that $n=\Omega(t^2)$ is large enough. We would denote by $k$ the number of needles in a stream in the planted model, where $k \sim \text{Bin}(t,p)$. We denote by $\unif$ the uniform distribution over $[n]^t$. 

First, we show how to prove lower bounds when $k$ is fixed. Given a $[t,k]$-interval system $F=\{I_1,\ldots,I_k\}$, we will assume in this section that the intervals are sorted in order, namely that $I_1$ comes before $I_2$, which comes before $I_3$, and so on. We define its corresponding sizes as
$$
\size(F) = (|I_1|,\ldots,|I_k|).
$$
We recall the definition of a planted stream distribution from the introduction, where we now present it more formally.

\begin{definition}[Planted distribution for interval systems]
Let $F$ be a $[t]$-interval system. we define a planted distribution $\planted[F]$ over streams $X \in [n]^t$ as follows:
\begin{enumerate}
    \item Sample uniform needle $x \in [n]$;
    \item In each interval $I \in F$ sample uniform index $a_I \in I$ and set $X_{a_I}=x$;
    \item For all $j \in [n] \setminus \{a_I: I \in F\}$, sample $X_j \in [n]$ uniformly.
\end{enumerate}

For $\cF$ a randomized $[t]$-interval system, we define its planted distribution $\planted[\cF]$ by first sampling $F \sim \cF$ and then $X \sim \planted[F]$.
\end{definition}

We start by formalizing and proving \Cref{lemma:intro:planted-interval-lb}. Given a streaming algorithm $\alg$ and two distributions $D_0,D_1$ over streams, we say that $\alg$ distinguishes between $D_0,D_1$ with error $\delta$ if, at the end of running the algorithm, the last player can guess if the input was sampled from $D_0$ or $D_1$ and be correct with probability at least $1-\delta$. A streaming algorithm is an $\ell$-pass streaming algorithm if it makes $\ell$ passes over the data stream.

\begin{lemma}
\label{lemma:embedding_protocol}
Let $F$ be a $[t,k]$-interval system and set $\s=\size(F)$. Let $\alg$ be an $\ell$-pass streaming algorithm which distinguishes between $\planted[F]$ and $\unif$ with error $0.5\%$ and uses space $s$.
Then there is a communication protocol $\Pi$ which solves the unique set-disjointness problem under input distribution $\dissize{}{\s}$, in which each player sends $\ell s$ bits, and has error $1\%$. 
\end{lemma}

\begin{proof}
Let $X=(X_1,\ldots,X_k) \in (\bits^n)^k$ be the input to the players, where we assume $X \sim \dissize{b}{\s}$ for some $b \in \{0,1\}$. The goal of the players is to figure out $b$. 

Let $F=\{I_1,\ldots,I_k\}$. Let $J_1,\ldots,J_k$ be a partition of $[t]$, where $I_i \subset J_i$. As a first step, each player individually constructs a stream $Y_i \in [n]^{J_i}$ based on their input $X_i$. The $i$-th player generates their stream as follows:
\begin{enumerate}
    \item For each $j \in J_i \setminus I_i$, sample $Y_i(j) \in [n]$ uniformly.
    \item Let $S_i=\{j \in [n]: X_i(j)=1\}$, where $|S_i|=s_i$ be assumption. Let $L_i \in [n]^{s_i}$ be a random permutation of $S_i$. Set $(Y_i(j): j \in I_i) = L_i$.
\end{enumerate}
Let $Y = Y_1 \circ \cdots \circ Y_k \in [n]^t$ be the concatenation of the streams. The players simulate running $\alg$ on the stream, where each player simulates it on their part of the stream, and send the internal memory of the streaming algorithm to the next player. At the end of each pass, the last player sends the internal memory back to the first player. Thus each player sends at most $\ell s$ bits. To conclude, we need to show that this allows to distinguish between $b=0$ and $b=1$.

To conclude, we compute the distribution of $Y$ based on the value of $b$, and show that when $b=0$ the distribution of $Y$ is close to uniform, and when $b=1$ it is close to the planted distribution $\planted[F]$. Thus by assumption the algorithm distinguishes between these two cases, which is our goal.

First, if $b=0$ then $X_1,\ldots,X_k$ are uniform sets of sizes $s_1,\ldots,s_k$ in $[n]$, conditioned on being pairwise disjoint. Thus the elements of $Y$ are uniform among all choices of $t$ distinct elements in $n$. Since we assume $n=\Omega(t^2)$, the statistical distance between $Y$ and $\unif$ is at most $t^2/n$, which can be made as small as we want, say $0.1\%$.

Similarly, if $b=1$ then $X_1,\ldots,X_k$ are uniform conditioned on having a unique intersection. Similarly, the assumption $n=\Omega(t^2)$ implies that the the statistical distance between $Y$ and $\planted[F]$ can be made as small as we want, say $0.1\%$.

Overall, as we assume that $\alg$ can distinguish between $\unif$ and $\planted[F]$ with error $0.5\%$, then it also distinguishes between the distributions of $Y$ for $b=0$ and $b=1$ with slightly larger error $1\%$.
\end{proof}

Combining \Cref{lemma:embedding_protocol} with \Cref{thm:udisj_size}, we obtain the following corollary which formalizes \Cref{lemma:intro:planted-interval-lb}.

\begin{lemma}
\label{lemma:planted-interval-lb}
Let $F$ be a valid $[t,k]$-interval system. Let $\alg$ be an $\ell$-pass streaming algorithm which distinguishes between $\planted[F]$ and $\unif$ with error $0.5\%$ and uses space $s$. Then
$$
\ell s = \Omega \left(\frac{1}{\val(F)} \right).
$$
\end{lemma}

\begin{proof}
Let $\Pi$ be the protocol obtained by \Cref{lemma:embedding_protocol}, which solves unique set-disjointness under inputs distribution $\dissize{}{\s}$ for $\s=\size(F)=[s_1,\ldots,s_k]$, and where each player sends at most $\ell s$ bits. Since $F$ is valid we have $\sum s_i \le t/2$. \Cref{thm:udisj_size} then gives
$$
\sum_{i \in [k]} \frac{\ell s}{s_i} = \Omega(1).
$$
Recalling the definition of $\val(F)=\sum_{i \in [k]} \frac{1}{s_i}$, we can rephrase this as $\ell s \cdot \val(F) = \Omega(1)$.
\end{proof}

The following lemma, which formalizes \Cref{lemma:intro:planted-rand-interval-lb}, generalizes \Cref{lemma:planted-interval-lb} to randomized interval systems.

\begin{lemma}
\label{lemma:planted-rand-interval-lb}
Let $\cF$ be a valid randomized $[t]$-interval system. Let $\alg$ be an $\ell$-pass streaming algorithm which distinguishes between $\planted[\cF]$ and $\unif$ with error $0.1\%$ and uses space $s$. Then
$$
\ell s = \Omega \left(\frac{1}{\val(\cF)} \right).
$$
\end{lemma}

\begin{proof}
Sample $F \sim \cF$. Since $\val(\cF) = \E[\val(F)]$, by Markov's inequality we have
$$
\Pr_F[\val(F) > 2 \val(\cF)] \le 50\%.
$$
Next, let $\text{err}(F)$ denote the error of $\alg$ in distinguishing $\planted[F]$ from $\unif$. Since $\planted[\cF]$ is a mixture of $\planted[F]$, then the average of $\text{err}(F)$ is the error of $\alg$ in distinguishing $\planted[\cF]$ from $\unif$, which we assume is $0.1\%$. Thus
$$
\Pr_F[\text{err}(F) > 0.5\%] \le 20\%.
$$
Overall, there is some choice of $F$ in the support of $\cF$ such that $\val(F) \le 2 \val(\cF)$ and $\text{err}(F) \le 0.5\%$. The lemma follows by applying \Cref{lemma:planted-interval-lb} to $F$.
\end{proof}

We now in place to finally prove \Cref{thm:intro:needle}, giving sample-space lower bounds for any streaming algorithm that solves the needle problem.

\begin{proof}[Proof of \Cref{thm:intro:needle}]
Let $\alg$ be an $\ell$-pass streaming algorithm which can distinguish with high probability between the uniform and planted needle distribution using $t$ samples. As the inputs are stochastic, we may repeat it a few times to decrease its error. Thus, by increasing $t$ by a constant multiplicative factor, we may assume that the error is at most $0.1 \%$ and that $t$ is a power of two. 

For $k \le t$ let $\cF_k$ be the valid perfect randomized $[t,k]$-interval system given by \Cref{lemma:interval_summary}. We construct a randomized $[t]$-interval system $\cF$ by sampling $k \sim \text{Bin}(t,p)$ and taking $\cF_k$. Observe that $\planted[\cF]$ is identical to the planted needle distribution. If $\alg$ uses $s$ bits of space then \Cref{lemma:planted-rand-interval-lb} gives that
$$
\ell s = \Omega \left(\frac{1}{\val(\cF)} \right).
$$
To conclude the proof we just need to compute $\val(\cF)$. For any fixed $k$ we have by \Cref{lemma:interval_summary} that
$$
\val(\cF_k) \le \frac{10 k^2 \log(t)}{t}.
$$
Since $k \sim \text{Bin}(t,p)$ we have $\E[k^2] = p(1-p)t+p^2 t^2$. Since we assume $p=\Omega(1/t)$, the dominant term is the quadratic term, and hence $\E[k^2]=\Theta(p^2 t^2)$. Thus we get
$$
\val(\cF) = O(p^2 t \log(t)).
$$
Rearranging the terms concludes the proof, since it gives $\ell p^2 s t \log(t) = \Omega(1)$.
\end{proof}

\section{Open problems}
\label{sec:open}

We proved in \Cref{thm:intro:needle} near-tight bound for the sample vs space complexity needed for the needle problem, which proves similar near-tight bounds for the frequency estimation in stochastic streams problem. It still remains open to prove sharp bounds, removing the remaining logarithmic factor. We propose the following natural conjecture.

\begin{conjecture}
Any $\ell$-pass streaming algorithm which can distinguish with high probability between the uniform and planted models, where $p$ is the needle probability, $t$ the number of samples, $s$ the space and $n$ the domain size, satisfies $\ell p^2 st = \Omega(1)$.   
\end{conjecture} 

Another natural conjecture is to remove the artificial restriction of $\sum s_i \le n/2$ from \Cref{thm:intro:udisj}. We need it because we do not prove the theorem directly, but rather via a reduction to the asymmetric product distribution case. We speculate that there may be a direct proof which overcomes this technical barrier (although we don't really have any application where the general bound is needed, it will be aesthetically pleasing to have a more complete result).

\bibliographystyle{alpha}
\bibliography{main.bib}

\end{document}